\documentclass[12pt,a4paper]{article}

\usepackage{amsfonts,amsmath,amssymb,amsthm}
\usepackage{latexsym,amscd}
\usepackage{amsbsy}
\usepackage{CJK}
\usepackage{indentfirst,mathrsfs}
\usepackage{amsfonts,amsmath,amssymb,amsthm}
\usepackage{latexsym,amscd}
\usepackage{array}
\usepackage{amsbsy}
\usepackage{color}
\usepackage{multirow}
\usepackage{makecell}
\usepackage{cases}

\numberwithin{equation}{section}
\newtheorem{thm}{Theorem}
\newtheorem{lem}{Lemma}
\newtheorem{cor}{Corollary}

\newtheorem{prop}{Proposition}

\newtheorem{defn}{Definition}
\newtheorem{example}{Example}
\newtheorem{remark}{Remark}

\newcommand{\Tr}{{\rm {Tr}}}

\begin{document}
\title{New  P$c$N and AP$c$N functions over finite fields}
\author{Yanan~Wu, Nian~Li,
        and~Xiangyong~Zeng
\thanks{The authors are with the Hubei Key Laboratory of Applied Mathematics, Faculty of Mathematics and Statistics, Hubei
University, Wuhan, 430062, China. Email: yanan.wu@aliyun.com, nian.li@hubu.edu.cn, xzeng@hubu.edu.cn}
}
\date{}
\maketitle
\begin{quote}
{\small {\bf Abstract:} Functions with low $c$-differential uniformity were proposed in $2020$ and attracted lots of attention, especially the P$c$N and AP$c$N functions, due to their applications in cryptography. The objective of this paper is to study   P$c$N and AP$c$N functions. As a consequence, we propose a class of P$c$N functions and  four classes of AP$c$N functions by using the cyclotomic technique and the switch method.   In addition, four classes of   P$c$N or AP$c$N functions are presented by virtue of (generalized) AGW criterion.
}


\end{quote}

\section{Introduction}
The differential attack introduced by Biham and Shamir in \cite{BS} has attracted a lot of attention since it is a powerful cryptanalytic method to attack block ciphers.  To measure the ability of a given function to resist the differential attack, Nyberg in \cite{N} introduced  the concept of differential $\delta$-uniformity:

\begin{defn} Let  $\mathbb{F}_q$  denote the finite field with $q$ elements. A function $F$  from $\mathbb{F}_q$ to itself is called differentially $\delta$-uniform, where $ \delta=\max\limits_{a \in \mathbb{F}_q^*,b \in \mathbb{F}_q} \#\left\{x \in \mathbb{F}_q: F(x+a)-F(x)=b\right\}.$
\end{defn}

If $\delta=1$, then $F$ is called a perfect nonlinear (PN) function or planar function. If $\delta=2$, then $F$ is called an almost perfect nonlinear (APN) function. It is well known that the lower  the quantity of $\delta$, the stronger the ability of this function to resist differential attack. Therefore, the functions with low differential uniformity have been widely investigated in the past years. Some known PN and APN functions can be found in \cite{BH,CHHKXZ,CM,DO,DY,D1,D2,D3,ZKW}.

Very recently, inspired from the development of a practical differential attack in \cite{BCJW}, Ellingsen, Felke, Riera, St$\breve{a}$nic$\breve{a}$ and Tkachenko in \cite{EFRST} proposed a new type of differential which is called multiplicative differential (and the corresponding $c$-differential uniformity) defined as follows:

\begin{defn} \label{def} Let  $\mathbb{F}_q$  denote the finite field with $q$ elements, where $q$ is a power of a prime $p$. Given a $p$-ary function $F: \mathbb{F}_{q} \rightarrow \mathbb{F}_{q}$ and $c \in \mathbb{F}_{q}\backslash\{1\}$, the (multiplicative) $c$-derivative of $F$ with respect to $a\in \mathbb{F}_{q}$ is defined as
$$
{}_cD_{a} F(x)=F(x+a)-c F(x).
$$
Let ${}_c\Delta_{F}= \max\limits_{a, b \in \mathbb{F}_q} \#\left\{x \in \mathbb{F}_q: {}_c D_{a} F(x)=b\right\}.$ Then $F$ is called the differentially $(c, {}_c\Delta_{F})$-uniform.
\end{defn}

Note that  if $c=0$ or $a=0,$ then ${}_c D_{a} F(x)$ is a shift of the function $F$ and if $c=1$ and $a\ne 0$, then ${}_c D_{a} F(x)$ becomes  the usual derivative. Correspondingly,  if  ${}_c \Delta_{F}=1$, then  $F$  is  called a perfect $c$-nonlinear (P$c$N) function. If  ${}_c \Delta_{F}=2,$ then  $F$ is called an almost perfect $c$-nonlinear (AP$c$N) function.
Note that  $F$ is a P$c$N function if and only if ${}_c D_{a} F(x)$ is a permutation polynomial over $\mathbb{F}_{q}$ for any $a\in\mathbb{F}_{q}$. In particular,  if $F$ is a permutation, then $F$ is a trivial P$c$N function for $c=0$.

Ellingsen et al.  in \cite{EFRST} also characterized $c$-differential uniformity of a function in terms of its Walsh transform and they further investigated  the $c$-differential uniformity of some well-known PN functions for odd characteristics. Besides, the inverse function 
was precisely characterized for both even and odd characteristics.
Thanks to the work in \cite{EFRST}, functions with low $c$-differential uniformity, especially the P$c$N and AP$c$N functions have attracted a lot of attention for their significant applications in cryptography. For instance,
Yan, Mesnager and Zhou in \cite{YMZ} proposed  some   power functions  over finite fields with low $c$-differential uniformity and some of them are P$c$N or AP$c$N. Very recently,  Bartoli and Calderini in \cite{BC} investigated the  existence of some AP$c$N and P$c$N functions and presented several classes of P$c$N and AP$c$N functions by using AGW criterion given by   Akbary,   Ghiocab and   Wang in \cite{AGW}. Another interesting problem is whether the $c$-differential uniformity
is preserved or not through affine, extended affine (EA) or CCZ-equivalence. With some detailed analysis, Hasan,  Pal,  Riera and St$\breve{a}$nic$\breve{a}$ in \cite{HPRS} concluded that $c$-differential uniformity is not invariant under EA-equivalence and  CCZ-equivalence. In addition,  by perturbing the P$c$N functions and the $p$-to-1 linearized polynomials, they presented some general framework of P$c$N functions. More relevant information can be found in \cite{RS,S,SG,SRT,ZH}.
%

To the best of our knowledge, only very few P$c$N and AP$c$N functions have been proposed and almost all of them are monomials. In this paper, we  aim to construct P$c$N and AP$c$N functions which are multinomials. Concretely, in Section $2$,  six classes of functions with low $c$-differential uniformity are presented. More precisely, by  using the cyclotomic technique and the switch method to deal with the $c$-differential equation ${}_c\Delta_F(x)=b$ defined in Definition \ref{def}, we obtain a class of P$c$N functions and four classes of AP$c$N functions.  As a byproduct,  a class of differential $(c, {}_c\Delta_F)$-uniform functions can be obtained by swapping any two points of the inverse function, where ${}_c\Delta_F\leq3$.
 Moreover, motivated by the work of Bartoli and Calderini \cite{BC}, in Section $3$, we propose four  classes of P$c$N or AP$c$N functions by virtue of AGW criterion and generalized AGW criterion shown in \cite{AGW} and \cite{MQ}, respectively.  In order to prove the existence of these functions,  the corresponding examples are also  given after each construction. Finally,  all the known P$c$N and AP$c$N functions are summarized in Table \ref{tab1} for comparison.

Hereafter, we always assume that $q$ is a power of a prime $p$ and  $\mathbb{F}_{q}$ denote the finite field with $q$ elements. Also, we denote $\Tr_1^q$ the absolute trace of the finite field $\mathbb{F}_q$ and $\Tr^{q^n}_q$ the trace function from $\mathbb{F}_{q^n}$ to $\mathbb{F}_q$.

\section{Several classes of  functions with low $c$-differential uniformity }\label{sec2}
In this section, we present six classes of  functions which are actually piecewise functions with low $c$-differential uniformity and almost all of them are P$c$N  or AP$c$N.
For  convenience,  we firstly state some basic notations and conclusions.

\begin{lem}{\rm(\cite{BRS})}\label{quadratic-solutions} Let $q=2^m$. If $a, b \in \mathbb{F}_{q}$ and $a\ne 0$, then
 the equation $x^{2}+a x+b=0$  has two solutions in $\mathbb{F}_{q}$ if and only if $\operatorname{Tr}_1^q\left(\frac{b}{a^{2}}\right)=0.$
\end{lem}

Let $l>1$ be a divisor of $q-1$ and $\omega$ be  a primitive element of $\mathbb{F}_{q}$. Define $D_i=\omega^i\langle\omega^l\rangle$ for $i=0,1,\cdots,l-1$, where $D_0$ is the multiplicative subgroup of  $\mathbb{F}_{q}$ generated by $\omega^l$. Then $\mathbb{F}_{q}=\{0\}\cup D_0\cup D_1\cup\cdots\cup D_{l-1}$. Note that $x^{\frac{q-1}{l}}=\omega^{\frac{q-1}{l}i}$ when $x\in D_i$ for $i=0,1,\cdots,l-1$.


\begin{thm}\label{f-cyclotomic} Let $l>1$ be a positive divisor of $q-1$ and $u\in\mathbb{F}_{q}$ with $u\ne 1,\,(1-l)\,\, {\rm mod}\,\, p$. Let $F(x)=x(\sum_{i=1}^{l-1}x^{\frac{q-1}{l}i}+u)\in\mathbb{F}_{q}[x]$.  If $c\in\mathbb{F}_{q}\backslash\{1\}$ satisfies either $1-\frac{l}{(1-c)(u+l-1)}$, $1+\frac{l}{(1-c)(u-1)}\in D_0$, or $1+\frac{cl}{(1-c)(u+l-1)}$, $1-\frac{cl}{(1-c)(u-1)}\in D_0$,
then ${}_c\Delta_F\leq2$.
\end{thm}

\begin{proof}We only consider the case  $1-\frac{l}{(1-c)(u+l-1)}$, $1+\frac{l}{(1-c)(u-1)}\in D_0$ since the other case can be similarly proved. In fact, $f(x)$ is a piecewise function with the form
\begin{equation}\label{f-piecewise}
F(x)=\left\{\begin{array}{ll}
0, & x=0; \\
(u+l-1)x, & x\in D_0;\\
(u-1)x,& x\in\cup^{l-1}_{i=1}D_i.
\end{array}\right.
\end{equation}
To complete the proof, we need to solve the equation
\begin{equation}\label{c-piecewise}
F(x+a)-cF(x)=b
\end{equation}
for any $(a,b)\in\mathbb{F}_q^2$. If $a=0$, then \eqref{c-piecewise} is reduced to $(1-c)F(x)=b$ which has at most two solutions since both $(u+l-1)x$ and  $(u-1)x$ have exactly one solution in $\mathbb{F}_q$. Next, we consider the case  $a\ne0$. Let $x\in\mathbb{F}_q\backslash\{0,-a\}$ be a solution of \eqref{c-piecewise} for a fixed $b\in\mathbb{F}_q$, then $x,\,x+a\ne 0$. We discuss  \eqref{c-piecewise} from the following two cases.

Case I. $x\in D_0$. Note that  $(1-c)(u+l-1)\ne0$ and $(1-c)(u+l-1)-l\ne0$ due to $c\ne 1,\, u\ne (1-l)\, {\rm mod} \,p$ and $1-\frac{l}{(1-c)(u+l-1)}\in D_0$.  Hence, we have

(i) if $x+a\in D_0$, then  $$x=\frac{b-a(u+l-1)}{(1-c)(u+l-1)}\,\,\,\,\mbox{and}\,\,\,\, x+a=\frac{b-ac(u+l-1)}{(1-c)(u+l-1)};$$

(ii) if $x+a\in\cup^{l-1}_{i=1}D_i$, then  $$x=\frac{b-a(u-1)}{(1-c)(u+l-1)-l}\,\,\,\,\mbox{and}\,\,\,\, x+a=\frac{b-ac(u+l-1)}{(1-c)(u+l-1)-l}$$ according to  \eqref{f-piecewise} and \eqref{c-piecewise}.

We claim that \eqref{c-piecewise}  has at most one solution in $D_0$. Suppose there exist two solutions $x_1$, $x_2\in D_0$ with $x_1+a\in D_0$ and $x_2+a\in \cup^{l-1}_{i=1}D_i$, then $\frac{x_1+a}{x_2+a}=1-\frac{l}{(1-c)(u+l-1)}\notin D_0$, which  contradicts with our conditions. Therefore, \eqref{c-piecewise}  has at most one solution in Case I.

Case II. $x\in\cup^{l-1}_{i=1}D_i$. Similar as Case I,  if $x+a\in D_0$, then
$$x=\frac{b-a(u+l-1)}{(1-c)(u-1)+l} \,\,\,\,\mbox{and}\,\,\,\,x+a=\frac{b-ac(u-1)}{(1-c)(u-1)+l}$$ and otherwise, one gets $$x=\frac{b-a(u-1)}{(1-c)(u-1)}\,\,\,\,\mbox{and}\,\,\,\,x+a=\frac{b-ac(u-1)}{(1-c)(u-1)}$$ from  \eqref{f-piecewise} and \eqref{c-piecewise}.  We then can claim that \eqref{c-piecewise} has at most one solution in $\cup^{l-1}_{i=1}D_i$ due to $1+l/(1-c)(u-1)\in D_0$.

 With the previous discussion, we conclude that  \eqref{c-piecewise} has at most two solutions in $\mathbb{F}_{q}\backslash\{0,\,-a\}.$ Besides, it can be easily proved that $x=0$ and $x=-a$ cannot be the solutions of \eqref{c-piecewise} simultaneously. Otherwise, if it has solutions 0 and $-a$, then we have $F(a)=b=cF(a)$ which leads to  $(1-c)F(a)=0$. This is impossible since $c\ne1$ and $F(a)\ne0$ for $a\ne0$.
 In the following, we discuss the number of solutions of \eqref{c-piecewise} when it has solution $x=0$ or $x=-a$.

 Suppose that $x=-a$ is a solution of  \eqref{c-piecewise}, i.e.,   $b=cF(a)$, then  \eqref{c-piecewise} has no solution in $D_0$ when $a\in D_0$ since $b=ac(u+l-1)$ and it has no solution in $\cup^{l-1}_{i=1}D_i$ when $a\in \cup^{l-1}_{i=1}D_i$ since $b=ac(u-1)$.

Suppose  $x=0$ is a solution of  \eqref{c-piecewise}, which implies  $b=F(a)$. One can see that when $a\in D_0$, then $b=a(u+l-1)$ which leads to \eqref{c-piecewise} has no solution in Case I(i). Moreover, $a\in D_0$ gives
 $$\frac{b-ac(u+l-1)}{(1-c)(u+l-1)-l}=\frac{a(1-c)(u+l-1)}{(1-c)(u+l-1)-l}\in D_0$$ due to  $1-\frac{l}{(1-c)(u+l-1)}\in D_0$ which  contradicts with $x+a\in\cup^{l-1}_{i=1}D_i$ in Case I(ii). Therefore, \eqref{c-piecewise} has no solution in Case I. Similarly, it has no solution in Case II when $a\in\cup^{l-1}_{i=1}D_i$.
 This completes the proof.
\end{proof}

\begin{example}Taking $q=5^4$, $l=4$ and $u=\omega$. Let  $c\in\mathbb{F}_{q}\backslash\{1\}$ satisfy $1-\frac{l}{(1-c)(u+l-1)}$, $1+\frac{l}{(1-c)(u-1)}\in D_0$, or $1+\frac{cl}{(1-c)(u+l-1)}$, $1-\frac{cl}{(1-c)(u-1)}\in D_0$, Magma shows that  the function $F(x)$ defined in Theorem \ref{f-cyclotomic} is AP$c$N.
\end{example}

\begin{example} Taking $q=2^8$, $l=3$ and $u=\omega$. Let  $c\in\mathbb{F}_{q}\backslash\{1\}$ satisfy $1-\frac{l}{(1-c)(u+l-1)}$, $1+\frac{l}{(1-c)(u-1)}\in D_0$, or $1+\frac{cl}{(1-c)(u+l-1)}$, $1-\frac{cl}{(1-c)(u-1)}\in D_0$, Magma shows that  the function $F(x)$ defined in Theorem \ref{f-cyclotomic} is P$c$N if $c=0$ and AP$c$N, otherwise.
\end{example}

Li, Helleseth and Tang in  \cite{LHT}  investigated a class of permutation polynomials with  the form of
\begin{equation}\label{x^p^k-x+L(x)}
f(x)=\left(x^{p^{k}}-x+\delta\right)^{s}+L(x)
\end{equation}
where $k,\, s$ are integers, $\delta \in \mathbb{F}_{p^m}$ and $L(x)$ is a linearized polynomial.
Motivated by their work, we study the $c$-differential uniformity of  the function with form of \eqref{x^p^k-x+L(x)}. As a result,   a class of AP$c$N functions is obtained as follows.

\begin{thm} \label{linear function-2} Let $m,\,k$ be positive integers with $m=3k$. Let $q=3^m$ and $F(x)=(x^{p^k}-x)^{\frac{q-1}{2}+1}+a_1x+a_2x^{p^k}+a_3x^{p^{2k}}$.   If $a_1,\,a_2,\,a_3\in \mathbb{F}_{3}$ with $a_1+a_2+a_3\neq 0$, then ${}_c\Delta_F\leq2$ for $c=-1$.
\end{thm}

\begin{proof} We only give the proof of the case $(a_1,a_2,a_3)=(0,1,1)$ since the others can be similarly proved. Let $(a_1,a_2,a_3)=(0,1,1)$ and $D_0=\langle \omega^2 \rangle$, then $F(x)$ can be expressed as
\begin{equation}\label{Linear-function-iii}
F(x)=\left\{\begin{array}{ll}
2x, & x^{p^k}-x=0;\\
2x+2x^{p^k}+x^{p^{2k}}, & x^{p^k}-x\in D_0;\\
x+x^{p^{2k}},& x^{p^k}-x\in D_1.
\end{array}\right.
\end{equation}
By Definition \ref{def}, for $c=-1$, it is sufficient to prove that the equation
\begin{equation}\label{c-differential-Linear-function-ii}
F(x+a)+F(x)=b
\end{equation}
has at most two solutions in $\mathbb{F}_q$ for any given $(a,b)\in \mathbb{F}_q^2$.

When $a\in\mathbb{F}_{p^k}$,  then \eqref{c-differential-Linear-function-ii} can be reduced to $F(x)=-b-a.$  Therefore, it is equivalent to proving $F(x)=b$ has at most two solutions for any $b\in\mathbb{F}_q$. (i) If $x^{p^k}-x=0$, from \eqref{Linear-function-iii}, one has $x=-b$ and $x^{p^k}-x=b-b^{p^k}$. (ii) If $x^{p^k}-x\in D_0$, then $F(x)=b$ can be reduced to $2x+2x^{p^k}+x^{p^{2k}}=b$. Adding this equation to its  $p^k$-th power equation, one can immediately obtain $x=b^{p^{2k}}+b$. Hence, in this case, one has $x^{p^k}-x=b^{p^k}-b^{p^{2k}}$. (iii) If $x^{p^k}-x\in D_1$, by combining equations $F(x)=b$, $F(x^{p^k})=b^{p^k}$ and $F(x^{p^{2k}})=b^{p^{2k}}$, one has $x=b^{p^{2k}}-b^{p^k}-b$ and $x^{p^k}-x=b^{p^{2k}}-b$. Observe that $(b-b^{p^k})^{p^{2k}}=(b^{p^k}-b^{p^{2k}})^{p^k}=b^{p^{2k}}-b$. Therefore, \eqref {c-differential-Linear-function-ii} has at most one solution when $a\in\mathbb{F}_{p^k}$.

When $a\notin\mathbb{F}_{p^k}$ and $a^{p^k}-a\in D_0$. One should note that $(x+a)^{p^k}-(x+a)=0$ indicates $x^{p^k}-x\in D_0$ when $-1$ is a square element in $\mathbb{F}_q$ and otherwise, $x^{p^k}-x\in D_1$. Without loss of generality,  suppose that $-1$ is a square element of $\mathbb{F}_q$ and $x$ is a solution to \eqref{c-differential-Linear-function-ii}, then we discuss \eqref{c-differential-Linear-function-ii} as follows:

Case I: $x^{p^k}-x=0$ and $(x+a)^{p^k}-(x+a)\in D_0$. One can directly calculate $x=b+a+a^{p^k}-a^{p^{2k}}$ and  $x^{p^k}-x=b^{p^k}-b-a^{p^{2k}}+a$.

Case II: $x^{p^k}-x\in D_0$ and $(x+a)^{p^k}-(x+a)=0$.  According to \eqref{Linear-function-iii} and \eqref{c-differential-Linear-function-ii}, one has $x+2x^{p^k}+x^{p^{2k}}=b+a$. Taking $p^k$-th power  on both sides of the equation gives $x^{p^k}+2x^{p^{2k}}+x=b^{p^k}+a^{p^k}$. Combining the two equations, one has $x=-b^{p^k}-b-a^{p^k}-a$ and $(x+a)^{p^k}-(x+a)=-b^{p^{2k}}+b+a^{p^{k}}-a^{p^{2k}}$.

Case III: $x^{p^k}-x\in D_0$ and $(x+a)^{p^k}-(x+a)\in D_0$. For this case, \eqref{c-differential-Linear-function-ii} is reduced to $x+x^{p^k}-x^{p^{2k}}=b+a+a^{p^k}-a^{p^{2k}}$. Raising both sides of the equation  to the power $p^k$, we have $-x+x^{p^k}+x^{p^{2k}}=b^{p^k}+a^{p^k}+a^{p^{2k}}-a$. These two equations lead to $x=-b^{p^{2k}}-b+a$ and $x^{p^k}-x=b^{p^{2k}}-b^{p^k}+a^{p^k}-a$.

Case IV: $x^{p^k}-x\in D_0$ and $(x+a)^{p^k}-(x+a)\in  D_1$. By \eqref{c-differential-Linear-function-ii}, one has $2x^{p^k}+2x^{p^{2k}}=b-a-a^{p^{2k}}$. Taking $p^k$-th power and $p^{2k}$-th power on both sides of the equation respectively, we have $2x^{p^{2k}}+2x=b^{p^k}-a^{p^k}-a$ and $2x+2x^{p^k}=b^{p^{2k}}-a^{p^{2k}}-a^{p^k}$. Combining the three equations gives $x=b^{p^{2k}}+b^{p^{k}}-b+a^{p^{k}}$. Therefore, $x^{p^k}-x=b^{p^{k}}-b+a^{p^{2k}}-a^{p^k}$ and $(x+a)^{p^k}-(x+a)=b^{p^{k}}-b+a^{p^{2k}}-a$.

Case V: $x^{p^k}-x\in  D_1$ and $(x+a)^{p^k}-(x+a)\in D_0$.  In this case, one has $2x^{p^k}+2x^{p^{2k}}=b+a+a^{p^k}-a^{p^{2k}}$. Similar as Case IV, one can obtain
 $x=b^{p^{2k}}+b^{p^{k}}-b-a^{p^k}-a$ and  $x^{p^k}-x=b^{p^{k}}-b-a^{p^{2k}}+a$.

Case VI: $x^{p^k}-x\in D_1$ and $(x+a)^{p^k}-(x+a)\in  D_1$. For this case, \eqref{c-differential-Linear-function-ii} is reduced to $2x+2x^{p^{2k}}=b-a-a^{p^{2k}}$. Performing as  before, we then have $x=b+b^{p^{k}}-b^{p^{2k}}+a$. Hence, we have $x^{p^k}-x=b-b^{p^{2k}}+a^{p^{k}}-a$ and $(x+a)^{p^k}-(x+a)=b-b^{p^{2k}}-a^{p^{k}}+a$.

Denote by $\Delta=b^{p^{2k}}-b^{p^k}+a^{p^k}-a$. We claim that the other cases cannot happen if Case I occurs. Indeed, if there is a solution in Case I, then we have $\Delta=0$ and $\Delta^{p^k}=0$ which indicate Cases III and V cannot happen. If \eqref{c-differential-Linear-function-ii} has a solution in Case II, then we can get $a^{p^k}=a$ due to $\Delta^{p^k}=0$ and $-b^{p^{2k}}+b+a^{p^{k}}-a^{p^{2k}}=0$. It  contradicts with $a\notin \mathbb{F}_{p^k}$. If Case IV or Case VI  happens, then $(x+a)^{p^k}-(x+a)=b^{p^{k}}-b+a^{p^{2k}}-a\in  D_1$ or $(x+a)^{p^k}-(x+a)=b-b^{p^{2k}}-a^{p^{k}}+a\in  D_1$, respectively.  This together with $\Delta^{p^{2k}}=0$, we always have $a-a^{p^{2k}}\in D_1$ which is impossible due to $a^{p^k}-a\in D_0$.

By a similar discussion as above, one can check that if Case II happens, then the other cases cannot happen. On the other hand,
If Case III occurs, then $\Delta\in D_0$ and $\Delta^{p^{2k}}=b^{p^{k}}-b-a^{p^{2k}}+a\in D_0$. It means that Case V can not happen. If case IV happens, then $x^{p^k}-x=b^{p^{k}}-b+a^{p^{2k}}-a^{p^k}\in D_0$ and $(x^{p^k}-x)^{p^{2k}}=b-b^{p^{2k}}+a^{p^{k}}-a\in D_0$. It implies that only one of Cases IV and VI can occur. From the previous discussions, we conclude that for any given $(a,b)\in \mathbb{F}_q^2$, if  $a^{p^k}-a\in D_0$, then at most two  of the above cases can occur simultaneously, i.e., \eqref{c-differential-Linear-function-ii} has at most two
solutions in $\mathbb{F}_q$.

When $a\notin\mathbb{F}_{p^k}$ and $a^{p^k}-a\in D_1$, the proof is similar to the case $a^{p^k}-a\in D_0$ and hence, we omit it here. This completes the proof.
\end{proof}

\begin{remark} Note that $F(x)^{p^{2k}}=(-1)^{\frac{q-1}{2}+1}(x^{p^{2k}}-x)^{\frac{q-1}{2}+1}+a_1x^{p^{2k}}+a_2x+a_3x^{p^{k}}$. Therefore, Theorem \ref{linear function-2} can be generalized as below:
Let $m,\,k,\,d$ be positive integers with $m=3d$ and $k=d$ or $k=2d$. Let $q=3^m$ and $F(x)=(x^{p^k}-x)^{\frac{q-1}{2}+p^{ik}}+a_1x+a_2x^{p^k}+a_3x^{p^{2k}}\in\mathbb{F}_{q}[x]$, where $0\leq i\leq 2$ is an integer. If $a_1,\,a_2,\,a_3\in \mathbb{F}_3$ with $a_1+a_2+a_3\ne 0$, then ${ }_{-1}\Delta_F\leq2$.
\end{remark}

\begin{example} Let $a_1=a_2=-1$ and $a_3=1$. Let $m=3d$ and $k=2d$ with $d=2$. Experiments show that $f(x)=(x^{p^k}-x)^{\frac{q-1}{2}+p^{2k}}-x-x^{p^k}+x^{p^{2k}}$ is an AP$c$N function over $\mathbb{F}_{3^m}$ if $c=-1$.
\end{example}

\begin{example} Let $a_1=a_2=1$ and $a_3=-1$. Let $m=3d$ and $k=2d$ with $d=1$. Experiments show that $f(x)=(x^{p^k}-x)^{\frac{q-1}{2}+p^{2k}}+x+x^{p^k}-x^{p^{2k}}$ is a P$c$N function over $\mathbb{F}_{3^m}$ for $c=-1$.
\end{example}

In what follows,  several classes of  P$c$N and AP$c$N functions are presented by virtue of the switch method. We firstly state a proposition as shown below.

\begin{prop} Let $q=2^m$ and $ \gamma\in\mathbb{F}_q^*$. Let $f(x)\in\mathbb{F}_q[x]$ be differentially $(c, \delta)$-uniform and $F(x)=f(x)\left(\Tr^q_1(x)+1\right)+f(x+\gamma)\Tr^q_1(x)$.
Then $F(x)$ is differentially $(c, \delta)$-uniform for $\Tr^q_1(\gamma)=0$, and otherwise, ${}_c\Delta_{F}\leq2\delta$.
\end{prop}

\begin{proof} Firstly, $F(x)$ can be expressed as
\begin{equation}
F(x)=\left\{\begin{array}{ll}
f(x), & \Tr^q_1(x)=0;\\
f(x+\gamma), & \Tr^q_1(x)=1.
\end{array}\right.
\end{equation}
Therefore, for the equation  $F(x+a)+cF(x)=b$, we have

Case I:  $a\in\mathbb{F}_q$ and $\Tr^q_1(a)=0$.
(i) When $\Tr^q_1(x)=0$, then it becomes $f(x+a)+cf(x)=b$;
(ii) when $\Tr^q_1(x)=1$, we have $f(x+\gamma+a)+cf(x+\gamma)=b$.

If $\Tr^q_1(\gamma)=0$, let $y=x+\gamma$, then (ii) is equivalent to $f(y+a)+cf(y)=b$ and $\Tr^q_1(y)=1$. Combining with (i), one can see that in this case, $F(x+a)+cF(x)=b$ is actually $f(x+a)+cf(x)=b$ with $\Tr^q_1(a)=0$.

If $\Tr^q_1(\gamma)=1$,  then  $x$ is a solution of (i) if and only if $x+\gamma$ is a solution of (ii).  Therefore, the equation has at most $2\delta$ solutions in this case since  $f(x)$ is $(c,\delta)$-differential uniformity.

Case II: $a\in\mathbb{F}_q$ and $\Tr^q_1(a)=1$. Then,
(i) when $\Tr^q_1(x)=0$,  $F(x+a)+cF(x)=b$ can be reduced to $f(x+a+\gamma)+cf(x)=b$;
(ii) when $\Tr^q_1(x)=1$,  one has $f(x+a)+cf(x+\gamma)=b$.

If $\Tr^q_1(\gamma)=0$, let $y=x+\gamma$, then (ii) is equivalent to $f(y+a+\gamma)+cf(y)=b$ and $\Tr^q_1(y)=1$. This together with (i),  $F(x+a)+cF(x)=b$ is actually $f(x+a)+cf(x)=b$ with $\Tr^q_1(a)=1$. Therefore, Cases I and II imply $F(x)$ is differentially $(c, \delta)$-uniform for $\Tr^q_1(\gamma)=0$.

The case $\Tr^q_1(\gamma)=1$ can be similarly proved. This completes the proof.
\end{proof}

With the help of the above proposition,  P$c$N and AP$c$N functions over finite fields with even characteristic can be obtained from known P$c$N functions.  Therefore, we will continue constructing  P$c$N and AP$c$N functions later.

Hasan et al. in \cite{HPRS}  characterized when the sum of a P$c$N and a Boolean function is also P$c$N as follows:

\begin{thm}\label{EA}\cite[Theorem 6.3]{HPRS} Let $1 \neq c \in \mathbb{F}_{q}$ be fixed, $G$ be a P$c$N function and $f$:  $\mathbb{F}_{q}\rightarrow\mathbb{F}_{p}$. Then $G+\gamma f$ is P$c$N if and only if the following conditions are satisfied:

(i) When $p=2$,  $\gamma$ is a $0$-linear structure of ${ }_{c} D_{a} f \circ\left({ }_{c} D_{a} G\right)^{-1}$ for all $a$.

(ii) When $p$ is odd,  for any $\lambda \in \mathbb{F}_{q}$ with $\operatorname{Tr}^q_1(\gamma \lambda)=\beta \in \mathbb{F}_{p}^{*},$
$$
\mathcal{W}_{R_{a}}(-\lambda, \beta)=\sum_{y \in \mathbb{F}_{p} n} \zeta^{\operatorname{Tr}\left(\beta R_{a}(y)+\lambda y\right)}=0
$$
where $\zeta$ is a $p$-root of unity, $R_{a}=H_{a} \circ \left({ }_{c} D_{a} G\right)^{-1}$ and $_{c}D_{a} f(x)=\operatorname{Tr}^q_1\left(H_{a}(x)\right)$.
\end{thm}

As mentioned in the end of \cite{HPRS}, it is  interesting to find  $G$ and $f$ such that $G+\gamma f$ is a P$c$N function. The next result is a direct consequence of Theorem \ref{EA}.
\begin{cor}\label{x+Tr(x^(p^k+1))} Let $k$, $m$ be any two integers and $d=\gcd(2^k+1,2^{2m}-1)$. Let  $q=2^m$ and  $F(x)=x+\gamma\Tr^{q^2}_1(x^{2^k+1})\in\mathbb{F}_{q^2}[x]$, where $\gamma\in\mathbb{F}_{q^2}$ and $\gamma^d=1$. If $c\in\mathbb{F}_{q^2}\backslash\{1\}$ satisfies either $c=0$ or $c^d=1$,  then $F(x)$ is a P$c$N function of $\mathbb{F}_{q^2}$.
\end{cor}

Next, we present a class of P$c$N functions which is not covered by Theorem \ref{EA}.

\begin{thm}\label{x+Tr(x)^(q-1)} Let $n$, $m$ be two integers and $q=p^m$ for any prime $p$. Let $L(x)$ be a linearized permutation polynomial over $\mathbb{F}_{q^n}$ with coefficients in $\mathbb{F}_{q}$ and  $F(x)=L(x)+L(\gamma)\Tr^{q^n}_q(x)^{q-1}\in \mathbb{F}_{q^n}[x]$, where $\gamma\in\mathbb{F}_{q^n}^*$ satisfies $\Tr^{q^n}_q(\gamma)=0$. Then $f(x)$ is a P$c$N function for any $c\in \mathbb{F}_{q^n}\backslash\{1\}$ and $\Tr^{q^n}_q\big(\frac{L(\gamma)}{1-c}\big)=0$.
\end{thm}

\begin{proof} Note that $\Tr^{q^n}_q(L(x))=L(\Tr^{q^n}_q(x))$ since $L(x)$ is a linearized  polynomial over $\mathbb{F}_{q^n}$ whose coefficients belong to $\mathbb{F}_{q}.$  By Definition \ref{def}, it suffices to prove the equation
\begin{equation}\label{x+y^(q-1)-cdifferentially}
F(x+a)-cF(x)=b
\end{equation}
has at most one solution for any $(a,b)\in\mathbb{F}_{q^n}^2$, where $c\in \mathbb{F}_{q^n}\backslash\{1\}$ with $\Tr^{q^n}_q\left(\frac{L(\gamma)}{1-c}\right)=0$. Note that
\begin{equation}
F(x)=\left\{\begin{array}{ll}
L(x), & \Tr^{q^n}_q(x)=0;\\
L(x+\gamma), & \Tr^{q^n}_q(x)\ne 0.
\end{array}\right.
\end{equation}
Therefore, we discuss \eqref{x+y^(q-1)-cdifferentially} as follows.

Case I: $\Tr^{q^n}_q(a)=0$. If $x$ is a solution of \eqref{x+y^(q-1)-cdifferentially} and $\Tr^{q^n}_q(x)=0$, then one has $F(x+a)-cF(x)=(1-c)L(x)+L(a)=b$ due to $\Tr^{q^n}_q(x+a)=0$. It leads to $L(x)=\frac{b-L(a)}{1-c}$ and $$L(\Tr^{q^n}_q(x))=\Tr^{q^n}_q(L(x))=\Tr^{q^n}_q\left(\frac{b-L(a)}{1-c}\right)=0.$$
If $x$ is a solution of \eqref{x+y^(q-1)-cdifferentially} with $\Tr^{q^n}_q(x)\neq0$, then $\Tr^{q^n}_q(x+a)\neq0$. In this subcase, \eqref{x+y^(q-1)-cdifferentially} can be reduced to $(1-c)L(x+\gamma)+L(a)=b$, which induces $L(x)=\frac{b-L(a)}{1-c}-L(\gamma)$ and $$L(\Tr^{q^n}_q(x))=\Tr^{q^n}_q(L(x))=\Tr^{q^n}_q\left(\frac{b-L(a)}{1-c}-L(\gamma)\right)=\Tr^{q^n}_q\left(\frac{b-L(a)}{1-c}\right)\neq 0$$
 due to $\Tr^{q^n}_q(L(\gamma))=L(\Tr^{q^n}_q(\gamma))=0$ and $L(x)$ being a linearized permutation polynomial. Therefore, one can see that \eqref{x+y^(q-1)-cdifferentially} has at most one solution in $\mathbb{F}_{q^n}$ since for any fixed $(a,b)\in\mathbb{F}_{q^n}^2$, one  has either
$\Tr^{q^n}_q\left(\frac{b-L(a)}{1-c}\right)=0$ or $\Tr^{q^n}_q\left(\frac{b-L(a)}{1-c}\right)\neq 0$.

Case II: $\Tr^{q^n}_q(a)\neq0$. If $\Tr^{q^n}_q(x)=0$, one then has  $\Tr^{q^n}_q(x+a)\neq0$, then \eqref{x+y^(q-1)-cdifferentially} is equivalent to $(1-c)L(x)+L(a+\gamma)=b$. Hence, one has $L(x)=\frac{b-L(a+\gamma)}{1-c}$ and
$$L(\Tr^{q^n}_q(x))=\Tr^{q^n}_q(L(x))=\Tr^{q^n}_q\left(\frac{b-L(a+\gamma)}{1-c}\right)= \Tr^{q^n}_q\left(\frac{b-L(a)}{1-c}\right)=0$$ due to $\Tr^{q^n}_q\big(\frac{L(\gamma)}{1-c}\big)=0$. If $\Tr^{q^n}_q(x)=-\Tr^{q^n}_q(a)\neq0$, then $\Tr^{q^n}_q(x+a)=0$. This leads to  $F(x+a)-cF(x)=(1-c)L(x)+L(a)-cL(\gamma)=b$. Thus, one has $L(x)=\frac{b-L(a)+cL(\gamma)}{1-c}$,
$$L(\Tr^{q^n}_q(x))=\Tr^{q^n}_q\Big(\frac{b-L(a)+cL(\gamma)}{1-c}\Big)= \Tr^{q^n}_q\Big(\frac{b-L(a)}{1-c}\Big)\neq0$$
and
$$L(\Tr^{q^n}_q(x+a))=\Tr^{q^n}_q(L(x)+L(a))=\Tr^{q^n}_q\Big(\frac{b-cL(a)+cL(\gamma)}{1-c}\Big)= \Tr^{q^n}_q\Big(\frac{b-cL(a)}{1-c}\Big)=0.$$
The above two equations hold due to $\Tr^{q^n}_q\left(\frac{cL(\gamma)}{1-c}\right)=\Tr^{q^n}_q\left(\frac{L(\gamma)}{1-c}-L(\gamma)\right)=0$.
If $\Tr^{q^n}_q(x)\neq0$ and $\Tr^{q^n}_q(x)\neq -\Tr^{q^n}_q(a)$, then \eqref{x+y^(q-1)-cdifferentially} is reduced to $(1-c)L(x+\gamma)+L(a)=b$. This induces $L(x)=\frac{b-L(a)}{1-c}-L(\gamma)$,
$$L(\Tr^{q^n}_q(x))=\Tr^{q^n}_q\left(\frac{b-L(a)}{1-c}-L(\gamma)\right)=\Tr^{q^n}_q\left(\frac{b-L(a)}{1-c}\right)\neq0$$ and $$L(\Tr^{q^n}_q(x+a))=\Tr^{q^n}_q\left(L(x)+L(a)\right)=\Tr^{q^n}_q\left(\frac{b-cL(a)}{1-c}\right)\neq0.$$ Similarly as Case I, one can conclude that \eqref{x+y^(q-1)-cdifferentially} has at most one solution in this case. This completes the proof.
\end{proof}

\begin{remark}  One should note that when $p$ is an odd prime, the condition $_{c}D_{a} f(x)=\operatorname{Tr}^q_1\left(H_{a}(x)\right)$ in Theorem \ref{EA} implies $c\in\mathbb{F}_{p}\backslash\{1\}$. Therefore,  $F(x)$ is a P$c$N function only for $c\in\mathbb{F}_{p}\backslash\{1\}$ in  Theorem \ref{EA} which
indicates our result is not covered by Theorem \ref{EA}.
\end{remark}

\begin{remark} Observe that  Proposition 1  holds only for even characteristic while Theorem \ref{x+Tr(x)^(q-1)} is considered for any characteristic $p$. On the other hand, let $f(x)=L(x)$ in Proposition 1, where $L(x)$ is a linearized permutation polynomials, then $F(x)=L(x)+L(\gamma)\Tr^q_1(x)$. One can check that this is a special case of Theorem \ref{x+Tr(x)^(q-1)} and otherwise, Theorem \ref{x+Tr(x)^(q-1)} and    Proposition 1 don't intersect.
\end{remark}

\begin{example} Selecting $q=5^2$ and $n=2$, Magma shows that for any $\gamma\in\mathbb{F}_{q^n}^*$ with $\Tr^{q^n}_q(\gamma)=0$,  $F(x)=x+\gamma\Tr^{q^n}_q(x)^{q-1}\in \mathbb{F}_{q^n}[x]$ is a P$c$N function if $c\in \mathbb{F}_{q^n}\backslash\{1\}$ satisfies $\Tr^{q^n}_q(\frac{\gamma}{1-c})=0$.
\end{example}

\begin{example} Let $q=3^3$ and $n=3$, experiments show that $F(x)=x^3+x+\gamma\Tr^{q^n}_q(x)^{q-1}\in \mathbb{F}_{q^n}[x]$ is a P$c$N function for any $c\in \mathbb{F}_{q^n}\backslash\{1\}$ and $\Tr^{q^n}_q(\frac{\gamma^3+\gamma}{1-c})=0$, where $\gamma\in\mathbb{F}_{q^n}^*$ with $\Tr^{q^n}_q(\gamma)=0$.
\end{example}

In the following, we  propose two classes of AP$c$N functions.

\begin{thm}\label{x^(p^k+1)+Tr(x)} Let $m$, $k$ be two positive integers with $\frac{m}{d}$ odd, where $d=\gcd(m,k)$.  Let $q=2^m$, $\gamma\in\mathbb{F}_{q}^*$ and $F(x)=x^{2^k+1}+\gamma\Tr^q_1(x)$. If $c\in\mathbb{F}_{2^d}\backslash\{1\}$, then ${}_c\Delta_F\leq2$.
\end{thm}

\begin{proof} Let $c\in\mathbb{F}_{2^d}\backslash\{1\}$, then $(1+c)^{2^k}=1+c$ and
 \begin{eqnarray*}
F(x+a)+cF(x)&=&(x+a)^{2^k+1}+\gamma\Tr^q_1(x+a)+c\left(x^{2^k+1}+\gamma\Tr^q_1(x)\right)\\
&=&(1+c)\Big(x^{2^k+1}+\frac{a}{1+c}x^{2^k}+\Big(\frac{a}{1+c}\Big)^{2^k}x\Big)+a^{2^k+1}\\&&+\gamma\left(\left(1+c\right)\Tr^q_1(x)+\Tr^q_1(a)\right)\\
&=&(1+c)\Big(x+\frac{a}{1+c}\Big)^{2^k+1}+\frac{ca^{2^k+1}}{1+c}+\gamma\left(\left(1+c\right)\Tr^q_1(x)+\Tr^q_1(a)\right).
 \end{eqnarray*}
  Since $\frac{m}{d}$ is odd, $\gcd(2^k+1,2^m-1)=1$. This leads to $F(x+a)+cF(x)=b$ has at most one solution for both $\Tr^q_1(x)=0$ and  $\Tr^q_1(x)=1$. Therefore, for any $(a,b)\in\mathbb{F}_{q}^2$,  $F(x+a)+cF(x)=b$ has at most two solution in $\mathbb{F}_{q}$.
\end{proof}

\begin{example}  Let $q=2^9$. For any $1\leq k\leq 8$, by a magma program,  $F(x)=x^{2^k+1}+\gamma\Tr^q_1(x)$ is an AP$c$N function of $\mathbb{F}_q$ for any $c\in\mathbb{F}_{2^d}\backslash\{1\}$, where $d=\gcd(m,k)$ and $\gamma\in\mathbb{F}_{q}^*$.
\end{example}

\begin{thm}
Let $m$ be an integer and $q=p^m$. Let $a_0$, $a_1\in\mathbb{F}_{q^2}$ with $a_1\neq a_0^{q}$. Then $F(x)=x^{q+1}+a_0x^q+a_1x$ is an AP$c$N function of $\mathbb{F}_{q^2}$ for any $c\in\mathbb{F}_{q}\backslash\{1\}$.
\end{thm}

\begin{proof} According to Definition \ref{def}, we consider the equation
\begin{equation}\label{x^(q+1)+L(x)-cdifferentially}
F(x+a)-cF(x)=b,
\end{equation}
where  $(a,b)\in\mathbb{F}_{q^2}^2$ and $c\in\mathbb{F}_{q}\backslash\{1\}$.
One can easily check that \eqref{x^(q+1)+L(x)-cdifferentially} can be reduced to
\begin{equation}\label{x^(q+1)+L(x)-cdifferentially-1}
(1-c)x^{q+1}+(a+a_0-ca_0)x^q+(a^q+a_1-ca_1)x=b-F(a).
\end{equation}
Let $x=y+t$, where $t=-\left(a_0+\frac{a}{1-c}\right)$. By the fact
  \begin{eqnarray*}
(1-c)t^q+a^q+a_1-ca_1&=&(1-c)\left(-a_0^q-\left(\frac{a}{1-c}\right)^q\right)+a^q+(1-c)a_1\\
&=&(1-c)(a_1-a_0^q)
  \end{eqnarray*}
  due to $c\in\mathbb{F}_{q}\backslash\{1\}$ and $a_1\neq a_0^{q}$, \eqref{x^(q+1)+L(x)-cdifferentially-1} is equivalent to
  $$(1-c)y^{q+1}+(1-c)(a_1^q-a_0)^qy=b^{'}$$
  where $b^{'}=b-F(a)-((1-c)(t+a_0)+a)t^q-((1-c)a_1+a^q)t$. Let $y=(a_1^q-a_0)z$, then the above equation can be reduced to
  \begin{equation}\label{x^(q+1)+L(x)-cdifferentially-2}
z^{q+1}+z=\frac{b^{'}}{(1-c)(a_1^q-a_0)^{q+1}}.
\end{equation}
By \cite[Theorems 8 and 10]{KCM}, one can  obtain \eqref{x^(q+1)+L(x)-cdifferentially-2} has at most two  solutions for any $(a,b)\in\mathbb{F}_{q^2}^2$. We claim that the number of the solutions to \eqref{x^(q+1)+L(x)-cdifferentially-2} can reach to $2$ for some $(a,b)\in\mathbb{F}_{q}^2$. For instance, selecting $b=F(a)+((1-c)(t+a_0)+a)t^q+((1-c)a_1+a^q)t$, then $b^{'}=0$. Right now,  \eqref{x^(q+1)+L(x)-cdifferentially-2} has solutions $x=0$ and $x=-1$. This completes the proof.
\end{proof}

\begin{example} Selecting $q=2^5$. Magma experiments show that if  $a_0$, $a_1\in\mathbb{F}_{q^2}$ and $a_1\neq a_0^{q}$,  then $F(x)=x^{q+1}+a_0x^q+a_1x$ is an AP$c$N function of $\mathbb{F}_{q^2}$ for any $c\in\mathbb{F}_{q}\backslash\{1\}$.
\end{example}

Inverse function, as a kind of well-known cryptographic function,  has been extensively studied for its differential  properties. By modifying two points of the inverse function, i.e., exchanging two image values of the inverse function, we can obtain a class of low $c$-differentially functions  as follows.

\begin{thm}\label{inverse function} Let $q=2^m$ and $t\in\mathbb{F}_q^*$. Let
\begin{equation}\label{inverse-piecewise}
F(x)=\left\{\begin{array}{ll}
0, & x=t;\\
t^{q-2}, & x=0; \\
x^{q-2},& otherwise.
\end{array}\right.
\end{equation}
If $c\in\mathbb{F}_q^*\backslash\{1\}$ and $\Tr^q_1(c)=\Tr^q_1(\frac{1}{c})=1$, then ${}_c\Delta_F\leq3$
\end{thm}

\begin{proof} According to Definition \ref{def}, we consider the equation
\begin{equation}\label{inverse-cdifferentially}
F(x+a)+cF(x)=b.
\end{equation}
 If $a=0$, it becomes $(1+c)F(x)=b$ which has exactly one solution since $F(x)$ is a permutation. Next, we always assume $a\ne0$.

Case I:  $a=t$.  If \eqref{inverse-cdifferentially} has the solution $x=0$, then $b=ct^{-1}$. If $x=t$ is a solution, we then have $b=t^{-1}$. Therefore, $0$ and $t$ can not be solutions of \eqref{inverse-cdifferentially} simultaneously due to $c\neq1$. Suppose $x\ne 0,\,t$ which implies that $x+t\ne 0,\,t$. Therefore, \eqref{inverse-cdifferentially} is equivalent to $bx^2+(bt+c+1)x+ct=0$. When $b=0$ or $b=t^{-1}(1+c)$, it has only one solution. When $b\ne0,\,t^{-1}(1+c)$, by Lemma \ref{quadratic-solutions}, the equation has two solutions if and only if $\Tr^q_1\big(\frac{ctb}{(bt+c+1)^2}\big)=0$. Taking $b=ct^{-1}\,({\rm or}\,\, b=t^{-1})$, then $\operatorname{Tr}_{1}^{q}\big(\frac{c t b}{(b t+c+1)^{2}}\big)=\operatorname{Tr}_{1}^{q}(c)=1\big({\rm or}\,\operatorname{Tr}_{1}^{q}\big(\frac{ctb}{(b t+c+1)^{2}}\big)=\operatorname{Tr}_{1}^{q}\big(\frac{1}{c}\big)=1\big)$.  Therefore, \eqref{inverse-cdifferentially} has at most two solutions in this case.

Case II: $a\ne t$. Firstly, if $x=0$ or $x=a+t$ is a solution of \eqref{inverse-cdifferentially}, then $b=\frac{ac+t}{at}$ or $b=\frac{c}{a+t}$, respectively. If $x=a$ is a solution of  \eqref{inverse-cdifferentially}, we have $b=\frac{a+ct}{at}$, and $x=t$ is a solution of \eqref{inverse-cdifferentially} implies that $b=\frac{1}{a+t}$.   Observe that if \eqref{inverse-cdifferentially} has solutions $x=0$ or $x=a+t$, then one has $\frac{ac+t}{at}=b=\frac{c}{a+t}$. This leads to $t^2+at+a^2c=0$ which is impossible due to $\Tr^q_1(c)=1$. Similarly, one can conclude that \eqref{inverse-cdifferentially} has no more than one solution in the set $\{0,\,a,\,t,\,a+t\}$. In the following, we assume that $x\notin\{0,\,a,\,t,\,a+t\}$. Then \eqref{inverse-cdifferentially} becomes $(x+a)^{-1}+cx^{-1}=b$ which is equivalent to $bx^2+(ab+c+1)x+ac=0$. When $b=0$ or $b=a^{-1}(1+c)$, the equation has  exactly one solution. When $b\ne0,\,a^{-1}(1+c)$, by Lemma \ref{quadratic-solutions}, the equation has two solutions if and only if $\Tr^q_1\big(\frac{acb}{(ab+c+1)^2}\big)=0$. Therefore, there exist at most three solutions of \eqref{inverse-cdifferentially} if $a\ne t$.
This completes the proof.
\end{proof}

\begin{remark} Very recently, St$\breve{a}$nic$\breve{a}$ \cite{S} studied the inverse function by swapping $0$ and $1$ in even characteristic. He proved that if $m=2,$ then ${}_c \Delta_F \leq 1 ;$ if $m=3,$ then ${}_c \Delta_F \leq 3$ and otherwise, ${}_c \Delta_F \leq 4$. While, by swapping $0$ and any $t\in\mathbb{F}_q^*$, we characterize the conditions of $c$ such that ${}_c \Delta_F \leq 3$ for any $m$.   Hence, our result in Theorem \ref{inverse function} is not contained by the conclusion in  \cite{S}.
\end{remark}

\section{P$c$N and AP$c$N functions from (generalized) AGW criterion }\label{sec3}
In this section, motivated by the work of Bartoli and Calderini in \cite{BC}, we propose several classes of P$c$N and AP$c$N polynomials.  Firstly, we need to introduce some known results.
\begin{lem}{\rm(\cite{AGW})}\label{AGW} Let $\phi(x)$ and $\psi(x)$ be two $\mathbb{F}_{q}$-linear polynomials over $\mathbb{F}_{q}$ seen as endomorphisms of $\mathbb{F}_{q^{n}}$. Let $g \in \mathbb{F}_{q^{n}}[x]$ and $h \in \mathbb{F}_{q^{n}}[x]$ such that $h\left(\psi\left(\mathbb{F}_{q^{n}}\right)\right) \subseteq \mathbb{F}_{q}^*.$ Then
$$
f(x)=h(\psi(x)) \phi(x)+g(\psi(x))
$$
is a permutation polynomial of $\mathbb{F}_{q^{n}}$ if and only if the following two conditions hold:
\begin{enumerate}
  \item [(1)] $\operatorname{ker}(\phi) \cap \operatorname{ker}(\psi)=\{0\};$ and
  \item [(2)] $h(x) \phi(x)+\psi(g(x))$ permutes $\psi\left(\mathbb{F}_{q^{n}}\right)$.
\end{enumerate}
\end{lem}

Lemma \ref{AGW} is the well-known AGW criterion in the additive case proposed by Akbary, Ghioca and Wang which is an useful method to construct permutation polynomials.
In 2019, Mesnager and Qu generalized AGW criterion to construct 2-to-1 mappings over finite fields as follows.

\begin{lem}{\rm(\cite{MQ})}\label{AGW-like} Let $q=2^{m}, \phi(x)$ and $\psi(x)$ be two $\mathbb{F}_{q}$-linear polynomials over $\mathbb{F}_{q}$ seen as an endomorphisms of the $\mathbb{F}_{q}$-module $\mathbb{F}_{q^{n}},$ and let $g, h \in \mathbb{F}_{q^{n}}[x]$ such that $h\left(\psi\left(\mathbb{F}_{q^{n}}\right)\right) \subseteq \mathbb{F}_{q}^{*}$. If
$\operatorname{ker}(\phi) \cap \operatorname{ker}(\psi)=\{0, \alpha\}$ for some $\alpha \in \mathbb{F}_{q^{n}}^{*},$ and
$\bar{f}(x)=h(x) \phi(x)+\psi(g(x))$ permutes $\psi\left(\mathbb{F}_{q^{n}}\right)$,
then
$$
f(x)=h(\psi(x)) \phi(x)+g(\psi(x))
$$
is 2-to-1 over $\mathbb{F}_{q^{n}}$.
\end{lem}

As immediate consequences,  some general framework of permutation polynomials and 2-to-1 polynomials have been provided in \cite{AGW} and \cite{MQ}, respectively. The $c$-differential uniformity of some of these polynomials has been investigated by Bartoli and Calderini in \cite{BC}. In what follows, we will study the  remaining known polynomials given in \cite{AGW}.


\begin{thm}Let $\phi$ be an $\mathbb{F}_{q}$-linear  polynomial over $\mathbb{F}_{q}$ seen as endomorphisms of $\mathbb{F}_{q^{n}}$ with $\phi(1)\not\equiv\,0\,({\rm mod}\,p)$ and $g(x) \in \mathbb{F}_{q^{n}}[x]$. Let $u \in \mathbb{F}_{q}^{*}$ and $c\in\mathbb{F}_{q}\backslash\{1\}$. Then $F(x)=u\phi(x)+g\left(\operatorname{Tr}^{q^n}_{q}(x)\right)^{q}-g\left(\operatorname{Tr}^{q^n}_{q}(x)\right)$ is a P$c$N polynomial of $\mathbb{F}_{q^{n}}$  if and only if $\operatorname{ker}(\phi) \cap \operatorname{ker}\left(\operatorname{Tr}^{q^n}_{q}\right)=\{0\}$.
\end{thm}

\begin{proof}Denote by $\psi(x)=\operatorname{Tr}^{q^n}_{q}(x)$. Then $F(x)$ is P$c$N if and only if
$$F(x+a)-cF(x)=u(1-c) \phi(x)+g\left(\psi(x+a)\right)^q-g\left(\psi(x+a)\right)-cg\left(\psi(x)\right)^q+cg\left(\psi(x)\right)+u\phi(a)$$
is a permutation polynomial for any $a\in\mathbb{F}_{q^n}$. Note that $\phi(x)$ permutes $\mathbb{F}_{q}$ due to $\phi(1)\not\equiv\,0\,({\rm mod}\,p)$. Therefore, by Lemma \ref{AGW} and the fact that $$\psi\left(g\left(x+\psi(a)\right)^q-g\left(x+\psi(a)\right)-cg\left(x\right)^q+cg\left(x\right)\right)=0$$ for $c\in\mathbb{F}_{q}\backslash\{1\}$, we have  $F(x+a)-cF(x)$ is a permutation polynomial if and only if $\operatorname{ker}(\phi) \cap \operatorname{ker}\left(\operatorname{Tr}^{q^n}_{q}\right)=\{0\}$.  This completes the proof.
\end{proof}

\begin{example} Choose $\phi(x)=x$. Obviously, $\phi(x)$ be an $\mathbb{F}_{q}$-linear permutation polynomial over $\mathbb{F}_{q},$ and $\operatorname{ker}(\phi) \cap \operatorname{ker}\left(\operatorname{Tr}^{q^n}_{q}\right)=\{0\}$. Magma shows that for any $u \in \mathbb{F}_{q}^{*}$, $c\in\mathbb{F}_{q}\backslash\{1\}$ and $g(x) \in \mathbb{F}_{q^{n}}[x]$, the polynomial $F(x)=ux+g\left(\operatorname{Tr}^{q^n}_{q}(x)\right)^{q}-g\left(\operatorname{Tr}^{q^n}_{q}(x)\right)$ is  P$c$N.
\end{example}

\begin{thm}\label{x^q-x} Let $g(x)\in \mathbb{F}_{q^{n}}[x]$ restricted to $\mathbb{F}_{q}$ induce a permutation of $\mathbb{F}_{q}$ and $u \in \mathbb{F}_{q}^{*}$. Then the polynomial $F(x)=u\left(x^{q}-x\right)+g\left(\operatorname{Tr}^{q^n}_{q}(x)\right)$ is P$c$N for any $c\in\mathbb{F}_{q}\backslash\{1\}$ if and only if $p\nmid n$.
\end{thm}
\begin{proof}For any $c\in\mathbb{F}_{q}\backslash\{1\}$, one has
$$F(x+a)-cF(x)=u(1-c)(x^q-x)+g\left(\operatorname{Tr}^{q^n}_{q}(x+a)\right)-cg\left(\operatorname{Tr}^{q^n}_{q}(x)\right)+u(a^q-a).$$ $F(x)$ is a P$c$N polynomial if and only if $F(x+a)-cF(x)$ permutes $\mathbb{F}_{q^n}$ for any $a\in\mathbb{F}_{q^n}$. Note that
$$u(1-c) \left(x^{q}-x\right)+\operatorname{Tr}^{q^n}_{q}\left(g(x)+\operatorname{Tr}^{q^n}_{q}(a)-cg(x)\right)=
n(1-c)g(x)+n\operatorname{Tr}^{q^n}_{q}(a)$$ when $x\in \mathbb{F}_{q}$ since  $g(x)$ induces a permutation of $\mathbb{F}_{q}$. Then the  result is  a consequence of Lemma \ref{AGW}.
\end{proof}

\begin{example} Let $p$ be a prime with $p\nmid n$  and $g(x)=x^p$. Then $F(x)$ defined in Theorem \ref{x^q-x} is a P$c$N polynomial for any $c\in\mathbb{F}_{q}\backslash\{1\}$.
\end{example}

Next, according to the AGW and the generalized AGW criteria, we propose another two  classes of permutation polynomials or 2-to-1 polynomials and we prove that  they are P$c$N or AP$c$N for any $c\in\mathbb{F}_q\backslash\{1\}$. Note that a P$c$N polynomial or an AP$c$N polynomial for $c=0$ is actually a permutation
polynomial or 2-to-1 polynomial. Therefore, we state our result as follows.

\begin{thm}
Let $u \in \mathbb{F}_{q}^{*}$, $c\in\mathbb{F}_{q}\backslash\{1\}$ and $n$ be a positive integer with $p|n$. Let $\phi(x)$ be an $\mathbb{F}_{q}$-linear  polynomial over $\mathbb{F}_{q}$ seen as endomorphisms an of $\mathbb{F}_{q^n}$ with $\phi(1)\not\equiv\,0\,({\rm mod}\,p)$ and $g(x)$ be any polynomial in $\mathbb{F}_{q^{n}}[x]$ such that $g\left(\mathbb{F}_{q}\right) \subseteq \mathbb{F}_{q}.$ Then $F(x)=u\phi(x)+g\left(\operatorname{Tr}^{q^n}_{q}(x)\right)$ is a P$c$N polynomial if and only if $\operatorname{ker}(\phi) \cap \operatorname{ker}\left(\operatorname{Tr}^{q^n}_{q}\right)=\{0\}$.
\end{thm}
\begin{proof}Let $c\in\mathbb{F}_{q}\backslash\{1\}$. Then $f(x)$ is P$c$N if and only if
$$F(x+a)-cF(x)=u(1-c) \phi(x)+g\left(\operatorname{Tr}^{q^n}_{q}(x+a)\right)-cg\left(\operatorname{Tr}^{q^n}_{q}(x)\right)+u\phi(a)$$
is a permutation polynomial for any $a\in\mathbb{F}_{q^n}$. From Lemma \ref{AGW}, we can see that $F(x+a)-cF(x)$ is a permutation polynomial if and only if
$$u(1-c) \phi(x)+\operatorname{Tr}_{q}^{q^{n}}\left(g\left(x+\operatorname{Tr}_{q}^{q^{n}}(a)\right)-cg(x)\right)$$ permutes $\mathbb{F}_{q}$.
Note that $\phi(x)$ permutes $\mathbb{F}_{q}$ due to $\phi(1)\not\equiv0\pmod{p}$ and $\operatorname{Tr}^{q^n}_{q}(g(x))=g(x)\operatorname{Tr}^{q^n}_{q}(1)=0$ for any $x\in\mathbb{F}_{q}$ due to $g\left(\mathbb{F}_{q}\right) \subseteq \mathbb{F}_{q}$ and $p|n$. Therefore, the result follows due to the fact that $\operatorname{Tr}_{q}^{q^{n}}\left(g\left(x+\operatorname{Tr}_{q}^{q^{n}}(a)\right)-cg(x)\right)=
g\left(x+\operatorname{Tr}_{q}^{q^{n}}(a)\right)\operatorname{Tr}_{q}^{q^{n}}(1)+cg(x)\operatorname{Tr}_{q}^{q^{n}}(1)=0$. This completes the proof.
\end{proof}

\begin{example} Let $p$ be an any prime, $m$, $n$ be two positive integers with $p|n$ and $q=p^m$. Taking $\phi(x)=x$, one can easily check that $\phi(x)$ is an $\mathbb{F}_{q}$-linear permutation polynomial over $\mathbb{F}_{q}$ and $\operatorname{ker}(\phi) \cap \operatorname{ker}\left(\operatorname{Tr}^{q^n}_{q}\right)=\{0\}$. If $g(x)\in \mathbb{F}_{q^{n}}[x]$ satisfying $g\left(\mathbb{F}_{q}\right) \subseteq \mathbb{F}_{q}$, then $F(x)=x+g\left(\operatorname{Tr}^{q^n}_{q}(x)\right)$ is a P$c$N polynomial for any $c\in\mathbb{F}_{q}\backslash\{1\}$.
\end{example}

\begin{thm}\label{phix-di} Let $u \in \mathbb{F}_{q}^{*}$, $c\in\mathbb{F}_{q}\backslash\{1\}$,$J=\left\{x^{q}-x: x \in \mathbb{F}_{q^{n}}\right\}$ and $d$ be a positive divisor of $q-1$.  Let $\phi(x)$ be an $\mathbb{F}_{q}$-linear polynomial over $\mathbb{F}_{q}$ seen as an endomorphisms of $\mathbb{F}_{q^{n}}$ and $g(x) \in \mathbb{F}_{q^{n}}[x]$. Then
$
F(x)=u\phi(x)+g\left(x^{q}-x\right)^{\left(q^{n}-1\right) /d}
$
is a P$c$N polynomial of $\mathbb{F}_{q^{n}}$ if and only if $\operatorname{ker}(\phi) \cap \mathbb{F}_{q}=\{0\}$ and $\phi(x)$ permutes $J$. Furthermore, if $q=2^m$, $\phi$ is 2-to-1 over $\mathbb{F}_{q}$ and $\phi(x)$ permutes $J$, then $F(x)$ is  AP$c$N.
\end{thm}

\begin{proof} We only need to prove the first part of our theorem as the latter can be similarly obtained. Let $\psi(x)=x^q-x$. Again by Lemma \ref{AGW}, we have that
$$F(x+a)-cF(x)=u(1-c)\phi(x)+g\left(x^{q}-x+a^q-a\right)^{\left(q^{n}-1\right) /d}-cg\left(x^{q}-x\right)^{\left(q^{n}-1\right) /d}+u\phi(a)$$
is a permutation polynomial if and only if $\operatorname{ker}(\phi) \cap \mathbb{F}_{q}=\{0\}$ and
$$u(1-c)\phi(x)+\psi(g\left(x+a^q-a\right)^{\left(q^{n}-1\right) /d}-cg\left(x\right)^{\left(q^{n}-1\right) /d})$$
 permutes $J$.  This completes the proof since $\psi(g\left(x+a^q-a\right)^{\left(q^{n}-1\right) /d}-cg\left(x\right)^{\left(q^{n}-1\right) /d})$ is identically $0$.
\end{proof}

With a similar proof, the above  theorem  can be generalized as follows:
\begin{cor}\label{AGW-cor} Let $u \in \mathbb{F}_{q}^{*}$, $c\in\mathbb{F}_{q}\backslash\{1\}$, $J=\left\{x^{q}-x: x \in \mathbb{F}_{q^{n}}\right\}$.  Let $\phi(x)$ be a $\mathbb{F}_{q}$-linear polynomial over $\mathbb{F}_{q}$ seen as endomorphisms of $\mathbb{F}_{q^{n}}$ and $g(x) \in \mathbb{F}_{q^{n}}[x]$. Let $d_1,d_2,\cdots,d_t\in\{1\leq d\leq q-1:d\,|\,q-1\}$. Then the polynomial
$$
F(x)=u\phi(x)+\sum_{i=1}^{t}g\left(x^{q}-x\right)^{\left(q^{n}-1\right) /d_i}
$$
is a P$c$N polynomial of $\mathbb{F}_{q^{n}}$ if and only if $\operatorname{ker}(\phi) \cap \mathbb{F}_{q}=\{0\}$ and $\phi(x)$ permutes $J$. Furthermore, if $q=2^m$, $\phi$ is 2-to-1 over $\mathbb{F}_{q}$ and $\phi(x)$ permutes $J$, then $F(x)$ is  AP$c$N.
\end{cor}

\begin{remark} Note that the result in \cite[Theorem 3.5]{BC} is an immediate consequence of  Corollary \ref{AGW-cor} when $d=q-1$ and $t=1$.  Therefore, our result covers the second part of the result in \cite[Theorem 3.5]{BC}.
\end{remark}

\begin{table}[!htb]
\footnotesize
\caption{The Known P$c$N and AP$c$N functions over finite fields}\label{tab1}
\begin{center}
\newcommand{\tabincell}[2]{\begin{tabular}{@{}#1@{}}#2\end{tabular}}
\renewcommand\arraystretch{1.1}
\begin{tabular}{|c|c|c|c|}
%
%
 \hline $p$  &$F(x)$                                                       &${}_c\Delta_F$      &Ref.\\  \hline\hline
        any  &$x^2$                                                        & $2$             &\cite{EFRST}\\ \hline

        any  &$x^{q-2}$                      &$2$              &\cite{EFRST}\\ \hline

        odd  &$x^{\frac{p^k+1}{2}}$                           &$1$               & \cite{EFRST},\cite{HPRS},\cite{YMZ}\\  \hline

%
        odd  &$x^{p^2-p+1}$                                          &$1$               & \cite{BT}\\  \hline

        odd  &$x^{p^k+1}$                 &$2$               & \cite{YMZ}\\  \hline

%
        $3$  &$x^{\frac{3^k+1}{2}}$                 &$2$               & \cite{YMZ}\\  \hline
        %
%
       odd  &$x^{dp^j}$  &1   &\cite{HPRS} \\\hline

         2 &$x^{2^{k}}+\alpha(1+c) \operatorname{Tr}^q_1(x)$  &1 &\cite{HPRS} \\\hline

       any  &$x^{p^{k}+1}+\gamma \operatorname{Tr}^q_1(x)$&1   &\cite{HPRS} \\\hline

        3  &\multirow{2}{*}{$x^{\frac{k(q-1)+2}{p^s+1}}$}             &\multirow{2}{*}{$1$}               &\multirow{2}{*}{\cite{ZH}}\\  \cline{1-1}
         5 &   &               &\\  \hline

%

        any  &\multirow{2}{*}{$b\phi(x)+\Tr^{q^n}_q(g(x^q-x))$, $b\phi(x)+g(x^q-x)^{\frac{q^n-1}{q-1}}$}     &$1$
        &\multirow{2}{*}{\cite{BC}} \\  \cline{1-1}\cline{3-3}

        2  &    &$2$               & \\  \hline

         any  &$b\phi(x)+g(x^q-x)^{s}$                                           &$1$               &\cite{BC}\\  \hline

%
%

        any    & $x\big(\sum_{i=1}^{l-1} x^{\frac{q-1}{l} i}+u\big)$  &$\leq2$ & Theorem 1 in this paper\\ \hline

        3    & $(x^{p^k}-x)^{\frac{q-1}{2}+p^{ik}}+a_1x+a_2x^{p^{k}}+a_3x^{p^{2k}}$&$\leq2$ & Theorem 2 in this paper\\ \hline

         2&   $x+\gamma\Tr^{q^2}_1(x^{2^k+1})$ &1  & Corollary 1 in this paper\\ \hline

        any      &$L(x)+L(\gamma)\Tr^{q^n}_q(x)^{q-1}$    &   $1$ & Theorem 4 in this paper\\ \hline


         2 &  $x^{p^k+1}+\gamma\Tr^{q}_1(x)$    &   $\leq2$ &Theorem 5 in this paper\\ \hline

         any & $x^{q+1}+a_{0} x^{q}+a_{1} x$     &   $2$ &Theorem 6 in this paper\\ \hline


   any   &$u\phi(x)+g\big(\operatorname{Tr}^{q^n}_{q}(x)\big)^{q}-g\big(\operatorname{Tr}^{q^n}_{q}(x)\big)$ & $1$& Theorem 8 in this paper\\  \hline

        any    & $u\left(x^{q}-x\right)+g\big(\operatorname{Tr}^{q^n}_{q}(x)\big)$ &$1$ & Theorem 9 in this paper\\ \hline
%

        any    &$u\phi(x)+g\big(\operatorname{Tr}^{q^n}_{q}(x)\big)$  &$1$ & Theorem 10 in this paper\\\hline

        any  & \multirow{2}{*}{$u\phi(x)+\sum_{i=1}^{t}g\left(x^{q}-x\right)^{\left(q^{n}-1\right) /d_i}$} &$1$ & \multirow{2}{*}{Corollary 2 in this paper}\\ \cline{1-1}\cline{3-3}

         2  &  &$2$ &\\
        \hline
\end{tabular}
\end{center}
\end{table}

\begin{example} An easy example of  $\phi(x)$ such that $\phi(x)$ is an $\mathbb{F}_{q}$-linear polynomial over $\mathbb{F}_{q}$ and $\phi(x)$ permutes $J$ is given by $\phi(x)=x$. Let $p$ be an odd prime and $g \in \mathbb{F}_{q^{n}}[x]$. Then $f(x)=x+\left(x^{q}-x\right)^{\left(q^{n}-1\right) /2}+\left(x^{q}-x\right)^{\left(q^{n}-1\right) /(q-1)}$ is a P$c$N polynomial for any  $c\in\mathbb{F}_{q}\backslash\{1\}$.
\end{example}

\begin{example} Let $q=2^4$  and $\phi(x)=x^2+x\in\mathbb{F}_{q^3}[x]$. One can check that $\phi(x)$ is 2-to-1 over $\mathbb{F}_{q}$ and $\phi(x)$ permutes $J=\{x^q-x: x \in \mathbb{F}_{q^n}\}$. Experiments show that  $F(x)=x^2+x+\left(x^{q}-x\right)^{\left(q^{3}-1\right) /3}+\left(x^{q}-x\right)^{\left(q^{3}-1\right) /5}$ is an AP$c$N polynomial for any  $c\in\mathbb{F}_{q}\backslash\{1\}$.
\end{example}

\section{Conclusion remarks}
In this paper, we mainly focused on the constructions of P$c$N and AP$c$N functions. Briefly, we presented five classes of P$c$N or AP$c$N functions by using the cyclotomic technique and the switch method. Moreover, by employing AGW and generalized AGW criteria, we proposed four classes of P$c$N or AP$c$N functions. To end this paper, we summarize the known P$c$N and AP$c$N functions for $c\neq 0$ in Table \ref{tab1}.
%
%

A natural question is whether the P$c$N and AP$c$N functions presented in this paper are new or not. Hasan, Pal, Riera, St$\breve{a}$nic$\breve{a}$ in \cite{HPRS} showed that the $c$-differential uniformity of a given function $F(x)$ is preserved through $F\circ L(x)$ for an affine permutation $L(x)$ (note that it is not preserved through $L_1\circ F\circ L_2(x)$ for affine permutations $L_1(x)$ and $L_2(x)$) and is not invariant under EA-equivalence and CCZ-equivalence. By comparing the algebraic degrees, the values of $c$ and the characteristics $p$, it can be readily verified that the functions constructed in Section \ref{sec2} are not equivalent to the known ones under the operation $F\circ L(x)$, and the relation of the functions constructed in Section \ref{sec3} with the known ones remains unknown due to the uncertainly of $\phi(x)$ and $g(x)$. As a future work, we will investigate the $c$-differential invariants and construct more P$c$N and AP$c$N functions from different approaches.

\end{document}